\documentclass[11pt,reqno]{amsart}
\usepackage{graphicx}
\usepackage{mathtools}
\usepackage{ulem}
\usepackage{color}

\theoremstyle{plain}

\newtheorem{thm}{Theorem}[section]

\newtheorem{prop}[thm]{Proposition}

\theoremstyle{remark}

\newtheorem{rem}[thm]{Remark}

\numberwithin{equation}{section}
\theoremstyle{definition}

\newcommand{\beq}{\begin{equation}}
\newcommand{\eeq}{\end{equation}}
\newcommand{\bee}{\begin{equation*}}
\newcommand{\eee}{\end{equation*}}

\def\si{{\sigma}}

\DeclareMathAlphabet{\doba}{U}{msb}{m}{n}

\let\tr\trace

\def\divergence{{\mathop{\rm div}}}

\let\<\langle
\let\>\rangle

\newcommand{\definedas}{\mathrel{\raise.095ex\hbox{\rm :}\mkern-5.2mu=}}

\begin{document}


\title[]{Construction of vacuum initial data by the conformally covariant split system}
\author{Patryk Mach}
\address[Patryk Mach]{Instytut Fizyki im.\ Mariana Smoluchowskiego, Uniwersytet Jagiello\'{n}ski, {\L}ojasiewicza 11, 30-348 Krak\'{o}w, Poland.}
\email{patryk.mach@uj.edu.pl}
\author{Yaohua Wang}
\address[Yaohua Wang]{School of Mathematics and Statistics, Henan University, Kaifeng, Henan 475004, China.}
\email{wangyaohua@henu.edu.cn}
\author{Naqing Xie}
\address[Naqing Xie]{School of Mathematical Sciences, Fudan
University, Shanghai 200433, China.}
\email{nqxie@fudan.edu.cn}

\begin{abstract}
Using the implicit function theorem, we prove existence of solutions of the so-called conformally covariant split system on compact 3-dimensional Riemannian manifolds. They give rise to non-Constant Mean Curvature (non-CMC) vacuum initial data for the Einstein equations. We investigate the conformally covariant split system defined on compact manifolds with or without boundaries. In the former case, the boundary corresponds to an apparent horizon in the constructed initial data. The case with a cosmological constant is then considered separately. Finally, to demonstrate the applicability of the conformal covariant split system in numerical studies, we provide numerical examples of solutions on manifolds $\mathbb S^1 \times \mathbb S^2$ and $\mathbb S^1 \times \mathbb T^2$.
\end{abstract}

%
%
%




\maketitle



\section{Introduction}

Let $(M,g)$ be a compact 3-dimensional Riemannian manifold. Let $\sigma$ denote a symmetric trace- and divergence-free (TT) tensor of type $(0,2)$, and let $\tau$ be a function on $M$. Lichnerowicz \cite{Li44} and Choquet-Bruhat and York \cite{CBY80} developed the so-called conformal method to transform \textit{seed data} $(M, g, \sigma, \tau)$ into \textit{initial data} satisfying Einstein constraint equations. Consider the following system of equations for a positive function $\phi$ and a one-form $W$:
\begin{subequations}\label{SD}
\begin{align}
-8 \Delta \phi + R_g \phi &= -\frac{2}{3} \tau^2 \phi^{5} + |\sigma + L_gW|_g^2 \phi^{-7}, \label{hamiltonian} \\
\Delta_L W &= \frac{2}{3} \phi^6 d\tau. \label{momentum}
\end{align}
\end{subequations}
Here $\Delta = \nabla_i \nabla^i$ and $R_g$ are the Laplacian and the scalar curvature computed with respect to metric $g$, and $\Delta_L W$ is defined as $\Delta_L W = \mathrm{div}_g (L_g W)$, where $L_g $ is the conformal Killing operator,
\begin{equation}
\label{LWdef}
(L_g W)_{ij}=\nabla_iW_j +\nabla_j W_i -\frac{2}{3}(\mbox{div}_gW)g_{ij}.
\end{equation}
Equation (\ref{hamiltonian}) is called the \textit{Lichnerowicz equation}, and equation (\ref{momentum}) is called the \textit{vector equation}. System (\ref{SD}) is referred to as the \textit{vacuum conformal constraints}. A dual to a form $W$ satisfying the equation $L_g W = 0$ is called a conformal Killing vector field.

Suppose that a pair $(\phi,W)$ solves the vacuum conformal constraints (\ref{SD}). Define $\widetilde g = \phi^4 g$, and $K = \frac{\tau}{3} \phi^{4} g + \phi^{-2} (\si + L_g W)$. Then the triple $(M, \widetilde g, K)$ becomes an initial data set satisfying the \textit{vacuum Einstein's constraints}
\begin{subequations}\label{originalconstraints}
\begin{align}
R_{\widetilde g} - |K|^2_{\widetilde g} + (\tr_{\widetilde g} K)^2 &= 0, \\
\divergence_{\widetilde g} K - d \tr_{\widetilde g} K &= 0.
\end{align}
\end{subequations}
Note that $\mathrm{tr}_{\widetilde g} K = \tau$. Choquet-Bruhat and Geroch \cite{CBG69} showed that such initial data give rise to a unique development.

Here we briefly review the current status of the study of the existence of the solutions of the vacuum conformal constraints for closed manifolds $(M,g)$. The case of constant $\tau$ is basically understood, cf.\ \cite{Isenberg}. Of course, many important results can be still obtained assuming a constant $\tau$. Examples include studies of multiplicity of solutions in the case with a positive cosmological constant \cite{premoselli, walter, Mach_Knopik}, or obtaining foliations of important spacetimes \cite{malec, beig}. In general, the case of non-constant $\tau$ remains still open. Some results are obtained when $d \tau /\tau$ or $\sigma$ are small, cf.\ \cite{IM92,IsenbergOMurchadha,ACI08,MaxwellNonCMC}. Results for rough initial data can be found in \cite{HNT2, Maxwellrough}. Interested readers may refer to a survey paper, for instance \cite{BartnikIsenberg}.

Recently, Dahl, Gicquaud, and Humbert proved the following criterion for the existence of solutions to Eqs.\ (\ref{SD}) \cite{DGH}. Assume that $(M,g)$ has no conformal Killing vector fields and that $\sigma \not \equiv 0$, if the Yamabe constant $Y(g) \geq 0$. Then, if the \textit{limit equation}
\begin{equation}
\label{eqLimit1}
\Delta_L W = \alpha \sqrt{\frac{2}{3}} |L_g W|_g \frac{d\tau}{\tau}
\end{equation}
has no nonzero solutions for all $\alpha \in (0,1]$, the vacuum conformal constraints (\ref{SD}) admit a solution $(\phi,W)$ with $\phi>0$. Moreover, they provided an example on the sphere $\mathbb S^3$ such that the limit equation \eqref{eqLimit1} does have a nontrivial solution for some $\alpha_0 \in (0,1]$ \cite[Prop.\ 1.6]{DGH}. Unfortunately, the result of Dahl, Gicquaud, and Humbert is not an alternative criterion. In fact, there also exists an example such that both the limit equation \eqref{eqLimit1} and the vacuum conformal constraints (\ref{SD}) have nontrivial solutions \cite[Prop. 3.10]{TCN}.

There is another way to construct vacuum initial data. It is sometimes referred to as `the conformally covariant split' or, historically, `Method B.
' (See \cite[Section 4.1]{BartnikIsenberg} or the original paper \cite{York}) In this case we are trying to find a positive function $\phi$ and a one-form $W$ satisfying the so-called `conformally covariant split system:'
\begin{subequations}\label{MB}
\begin{align}
\Delta \phi -\frac{1}{8}R_g \phi + \frac{1}{8}|\si|_g^2\phi^{-7} + \frac{1}{4} \langle \si, L_g W \rangle_g \phi^{-1} - \left( \frac{1}{12}\tau^2-\frac{1}{8} |L_g W|_g^2 \right) \phi^{5} &= 0, \label{hamiltonian2} \\
\nabla_i(L_g W)^i_j-\frac{2}{3} \nabla_j\tau+6(L_g W)^i_j\nabla_i\log \phi &= 0. \label{momentum2}
\end{align}
\end{subequations}
Here $\sigma\in W^{1,p}$ is a symmetric TT-tensor of type $(0,2)$, $\tau \in W^{1,p}$ is a function on $M$, and $g \in W^{2,p}$ is a metric on $M$. We assume that $p > 3$. We devote this paper to the study of system (\ref{MB}).

Let $\tilde g = \phi^4 g$, and
\[ K = \frac{\tau}{3} \phi^{4} g + \phi^{-2} \si + \phi^{4} L_g W. \]

\begin{prop}
For $(\phi, W)$ solving system \eqref{MB}, the triple $(M, \widetilde g, K)$ becomes vacuum initial data.
\end{prop}

\begin{proof}
Denote $\widetilde{g}=\phi^{4} g$ and $K=\frac{\tau}{3} \phi^{4} g + \phi^{-2} \si + \phi^{4} L_g W$.
Combining the equations
\[ R_{\widetilde g}=\phi^{-5}(R_g \phi -8\Delta \phi), \]
\[ \tr_{\widetilde{g}}K= \frac{\tau}{3} \phi^{4} \tr_{\widetilde{g}}g+\phi^{-2} \tr_{\widetilde{g}}\si+\phi^{4} \tr_{\widetilde{g}}L_g W=\tau, \]
and
\[ |K|_{\widetilde{g}}^2=\frac{1}{3}\tau^2+\phi^{-12}|\si|_g^2 + |L_g W|_g^2 +2\phi^{-6} \langle \si,L_g W \rangle_g, \]
one gets
\begin{align*}
\MoveEqLeft R_{\widetilde g}-|K|_{\widetilde g}^2+(\tr_{\widetilde{g}}K)^2 \\
&= \phi^{-5}(R_g \phi -8\Delta \phi) -\left(\frac{1}{3}\tau^2+\phi^{-12}|\si|_g^2 + |L_g W|_g^2 +2\phi^{-6} \langle \si, L_g W \rangle_g \right)+\tau^2.
\end{align*}
Equation \eqref{hamiltonian2} implies $R_{\widetilde g}-|K|_{\widetilde{g}}^2+(\tr_{\widetilde{g}}K)^2=0$. On the other hand, by virtue of the fact that $\divergence_{\widetilde{g}}(\phi^{-2}\gamma)=\phi^{-6}\divergence_g(\gamma)$ for any trace-free symmetric $(0,2)$-tensor $\gamma_{ij}$ \cite[Eqn. ($\mathcal{C}$)]{De}, we have
\begin{align*}
\widetilde{\nabla}^iK_{ij} &= \widetilde{\nabla}^i \left[ \frac{\tau}{3} \widetilde{g}_{ij} + \phi^{-2} \si_{ij} + \phi^{4} (L_g W)_{ij} \right]\\
&= \frac{\widetilde{\nabla}^i\tau}{3}\widetilde{g}_{ij}+\phi^{-6}\nabla^i\si_{ij} +\phi^{-6}\nabla^i \left[ \phi^{6} (L_g W)_{ij} \right].
\end{align*}
Thus
\[ \widetilde{\nabla}^iK_{ij}-\widetilde{\nabla}_j\tr_{\widetilde{g}}K=  \nabla^i(LW)_{ij}+ 6\phi^{-1} (\nabla^i \phi)(L_g W)_{ij}-\frac{2 \nabla_j \tau}{3}. \]
Equation \eqref{momentum2} implies $\divergence_{\widetilde{g}} K - d \tr_{\widetilde{g}} K =0$.
\end{proof}

\begin{rem}\label{scaling}
An immediate observation concerning system (\ref{MB}) is that it admits the following scaling symmetry. Suppose that system (\ref{MB}) has a solution $(\phi, W)$. Set $\hat \phi = \mu^{-\frac{1}{4}} \phi$, $\hat W = \mu^{\frac{1}{2}} W$ for some positive number $\mu \in \mathbb R^+$. Then $(\hat \phi, \hat W)$ satisfy system (\ref{MB}) with the data $\hat \sigma$ and $\hat \tau$ given by $\hat \sigma_{ij} = \mu^{-1} \sigma_{ij}$, $\hat \tau = \mu^{\frac{1}{2}} \tau$.
\end{rem}

Method B and some of its variants were discussed in \cite{De}. System (\ref{MB}) possess the conformal covariance property in the following sense. Let $(\phi,W)$ solve system (\ref{MB}). Define $\omega = \omega(\sigma,\phi,W,g) := |\sigma+\phi^6 L_g W|_g$ and $P_{g,\omega} \phi := 8\Delta_g\phi - R_g \phi + \omega^2 \phi^{-7}$. Then the Lichnerowicz equation (\ref{hamiltonian2}) can be written as
\begin{equation}
\phi^{-5}P_{g,\omega}\phi=\frac{2}{3}\tau^2,
\end{equation}
and the vector equation (\ref{momentum2}) becomes
\begin{equation}
\Delta_{g,\phi}W=\frac{2}{3}d\tau,
\end{equation}
where $\Delta_{g,\phi} W = \phi^{-6} \mathrm{div}_g (\phi^6 L_g W)$. Now we make the following conformal change:
\begin{equation*}
\widetilde{g}=\psi^4g, \quad \widetilde{\phi}=\psi^{-1}\phi, \quad \widetilde{\sigma}=\psi^{-2}\sigma, \quad \widetilde{W}=\psi^4 W,
\end{equation*}
where $\psi$ is any positive function. It is easy to show that $\widetilde{\sigma}$ is still TT with respect to the metric $\widetilde{g}$, $\widetilde{\omega} = \widetilde{\omega} (\widetilde{\sigma},\widetilde{\phi},\widetilde{W},\widetilde{g}) = \psi^{-6} \omega$, and $\Delta_{\widetilde{g},\widetilde{\phi}} \widetilde{W} = \Delta_{g,\phi}W$. Then the operator given by
\begin{equation*}
\mathcal{P}_g \begin{pmatrix} \phi \\ W \\ \end{pmatrix}
:= \begin{pmatrix} \phi^{-5}P_{g,\omega} \phi \\ \Delta_{g,\phi} W \end{pmatrix}
\end{equation*}
is conformally covariant, i.e.,
\begin{equation}
\mathcal{P}_{\widetilde{g}} \begin{pmatrix} \widetilde{\phi} \\ \widetilde{W} \end{pmatrix}
= \mathcal{P}_{g} \begin{pmatrix} \phi \\ W \end{pmatrix}
= \frac{2}{3} \begin{pmatrix} \tau^2 \\ d\tau \end{pmatrix}.
\end{equation}

If $\tau$ is constant, Eqs.\ (\ref{SD}) split in a natural way. In this case, we have $W\equiv 0$, and we are only left with the well-studied Lichnerowicz equation. Much less is mathematically known about the conformally covariant split system, although it was applied in certain studies by numerical relativists \cite{Co00}. The proof or disproof of the existence of solutions of system (\ref{MB}) is difficult. When $\tau$ is constant, $L_g W \equiv 0$ is of course a trivial solution of \eqref{momentum2}, but even in this case the existence of solutions with non-zero $L_g W$ is unclear.

Solutions of systems (\ref{SD}) (standard conformal method) and (\ref{MB}) (conformally covariant split system) are, of course, related. Suppose that system (\ref{SD}) has a solution $(\phi,W)$ for the assumed data $(M,g)$, $\sigma$, and $\tau$. Suppose further that $\hat W$ is a solution to the equation
\begin{equation}
\label{bothsystemsrelated}
L_g \hat W = \phi^{-6} L_g W.
\end{equation}
It can be easily checked that the pair $(\phi, \hat W)$ solves system (\ref{MB}) with the same assumed data $(M,g)$, $\sigma$, and $\tau$. Moreover, both solutions lead to the same $K$, and hence to the same initial data $(M, \tilde g, K)$. The subtlety of the relation between $W$ and $\hat W$, and hence between systems (\ref{SD}) and (\ref{MB}), is due to the fact that $L$ is not, in general, invertible. However, computing the divergence of both sides of Eq.\ (\ref{bothsystemsrelated}) we obtain
\[ \Delta_L \hat W = \mathrm{div}_g \left( \phi^{-6} L_g W \right). \]
Note that $L_g^* = - 2 \, \mathrm{div}_g$ is a formal adjoint of $L_g$ with respect to $L^2$ product. Thus $\Delta_L = - \frac{1}{2} L_g^\ast L_g$. In explicit terms
\[ \Delta_L W_j = \nabla^i (L_g W)_{ij} = \nabla^i \nabla_i W_j + \frac{1}{3} \nabla_j (\nabla^i W_i) + W_i R^i_j, \]
where $R^i_j$ are the components of the Ricci tensor. When $(M,g)$ has no conformal Killing vector fields, the vector Laplacian $\Delta_L = - \frac{1}{2} L_g^\ast L_g$ is bijective between certain Sobolev spaces. In this case, one can solve (\ref{bothsystemsrelated}) to obtain
\begin{equation*}
\hat W = \Delta_L^{-1}\left[ \mathrm{div}_g \left( \phi^{-6}L_g W \right) \right], \quad \mbox{or equivalently} \quad W = \Delta_L^{-1} \left[ \mathrm{div}_g \left( \phi^{6} L_g \hat W \right) \right].
\end{equation*}

Suppose that we already have vacuum initial data $(M,g,K)$ such that $\bar \tau = \tr_g K$ is constant. In this case the traceless part of $K$, $\bar\si_{ij} = K_{ij} - \frac{\tr_g K}{3}g_{ij}$ is divergence free, and  system \eqref{MB} admits a particular solution $(\bar\phi \equiv1, \bar W \equiv 0)$ for the data $(\sigma, \tau) = (\bar \sigma, \bar \tau)$. This obvious solution can be understood as transforming the seed data $(M,g,K)$ into itself. In subsequent sections, we use the implicit function theorem to deduce existence of other solutions of Eqs.\ (\ref{MB}) with $\tau \neq \tr_g K$.

The order of this paper is as follows. In Section \ref{closed} we prove existence of solutions of system (\ref{MB}) on closed manifolds $M$, but admitting a non-constant $\tau$. In Section \ref{boundarycase} we obtain similar results for a compact manifold with a boundary. The assumed boundary conditions guarantee that this boundary constitutes an apparent horizon in the obtained initial data. To a large extent, the results of Sec.\ \ref{closed} and \ref{boundarycase} constitute counterparts to the analysis given in \cite{GN14} for system (\ref{SD}) (Method A). In Section \ref{withlambda} we take into account a cosmological constant and reformulate main theorems obtained in Sec.\ \ref{closed}. Finally, Sec.\ \ref{numerics} provides simple numerical examples of solutions of system (\ref{MB}) on compact manifolds, assuming non-constant $\tau$. In these examples, we assume the manifold $M = \mathbb S^1 \times \mathbb S^2$ or $M = \mathbb S^1 \times \mathbb T^2$.

We use the standard geometric notation in this paper. The product of two symmetric tensors of type (0,2) with respect to metric $g$ is expressed, in terms of their tensor components, as $\langle a, b \rangle_g = g^{ij} g^{kl} a_{jl} b_{ik}$. The square of the norm is denoted as $|a|^2_g = \langle a, a \rangle_g$. We omit the subscript referring to the metric, when the metric is obvious from the context. In most cases the metric is understood to be that of the seed data. This also applies to the standard notation of the trace, or the divergence of a tensor, as well as the scalar curvature $R_g$ and the conformal Killing operator $L_g$.

Unless otherwise stated, all the given data on the manifold $M$ are assumed to be smooth. We make use of the standard notation $W^{k,p}$ to denote the Sobolev space of functions defined on the manifold $M$ and $W^{k,p}_+$ to denote the subset consisting of all positive $W^{k,p}$ functions. In Section \ref{boundarycase} the manifold $(M,g)$ is assumed to be compact with boundary. To distinguish the functional spaces on the manifold $M$ and those on the boundary $\partial M$, we use the notation  $W^{k,p}(M)$ and $W^{k,p}(\partial M)$ respectively.


\section{Conformally covariant split system on a closed manifold}
\label{closed}

In this section, we assume that $(M,g)$ is a closed 3-dimensional Riemannian manifold. Making use of the implicit function theorem, we construct a family of solutions of the conformally covariant split system \eqref{MB} on $M$. These solutions give rise to vacuum initial data.

\begin{thm}\label{3}
Suppose that we already have vacuum initial data $(M,g,K)$. Assume that $\bar \tau = \tr_g K = \mathrm{const}$, and that $K \neq 0$ in some region of $M$. Assume further that $(M,g)$ has no conformal Killing vector fields. Then there is a small neighborhood of $\bar \tau$ in $W^{1,p}$ such that for any $\tau$ in this neighborhood there exists $(\phi_\tau, W_\tau) \in W^{2,p}_+\times W^{2,p}$ solving the system \eqref{MB} for the data $\bar\si_{ij} = K_{ij} - \frac{\bar \tau}{3} g_{ij}$ and $\tau$.
\end{thm}

\begin{proof}
The proof is based on the implicit function theorem and the ideas are borrowed from \cite{Ch04} and \cite{GN14}. First, let us define the operator
\begin{align*}
\MoveEqLeft \mathcal{F} \colon W^{1,p} \times W^{2,p}_+\times W^{2,p} \rightarrow L^p\times L^p, \\
& \begin{pmatrix} \tau\\ \phi\\ W \end{pmatrix} \mapsto \begin{pmatrix}
\Delta \phi -\frac{1}{8}R \phi + \frac{1}{8}|\bar\si|^2\phi^{-7}+\frac{1}{4} \langle \bar\si, LW \rangle \phi^{-1} - \left( \frac{1}{12}\tau^2-\frac{1}{8}|LW|^2 \right) \phi^{5} \\
\nabla_i(LW)^i_j-\frac{2}{3}
\nabla_j \tau + 6(LW)^i_j\nabla_i\log \phi
\end{pmatrix}.
\end{align*}
It is easy to see that $\mathcal{F}$ is a $C^1$-mapping and $\mathcal{F}(\bar\tau,\bar \phi\equiv1,\bar W\equiv0)=(0,0)$. We prove that the partial derivative of $\mathcal{F}$ with respect to the variables $(\phi,W)$ is an isomorphism at $(\bar\tau,\bar \phi\equiv1,\bar W\equiv0)$. The differential at the point $(\bar\tau,\bar \phi\equiv1,\bar W\equiv0)$ is given by
\begin{align*}
\MoveEqLeft D\mathcal{F}|_{(\bar\tau,1,0)} \begin{pmatrix} \delta\phi \\ \delta W \end{pmatrix} \\
&= \begin{pmatrix}
\Delta-\frac{1}{8}R-\frac{7}{8}|\bar\sigma|^2-\frac{5}{12}\bar\tau^2 & , & \frac{1}{4} \langle \bar\sigma, L(\cdot) \rangle \\
 0 & , & \Delta_L
\end{pmatrix}
\begin{pmatrix}\delta\phi \\ \delta W\end{pmatrix},
\end{align*}
and it is triangular, meaning that the second row of the above $2 \times 2$ block matrix does not depend on $\delta \phi$. Thus, the invertibility of $D\mathcal{F}|_{(\bar\tau,1,0)}$ follows from the fact that the diagonal terms are invertible. More specifically:

\textit{Claim 1.}
\begin{align*}
\MoveEqLeft \mathcal{H} \colon W^{2,p} \rightarrow L^p, \\
& \delta\phi \mapsto (\Delta-\frac{1}{8}R-\frac{7}{8}|\bar\sigma|^2-\frac{5}{12}\bar\tau^2 )\delta \phi
\end{align*}
is invertible and

\textit{Claim 2.}
\begin{align*}
\MoveEqLeft \Delta_L \colon W^{2,p} \rightarrow L^p, \\
& \delta W \mapsto \Delta_L \delta W
\end{align*}
is also invertible.

The proof of Claim 2 is a consequence of the assumption that $(M,g)$ is closed and has no conformal Killing vector fields. The proof of Claim 1 is as follows. Note that $\mathcal{H}$ is a Fredholm operator of index $0$. It suffices to show that $\mathcal{H}$ is injective. Since $(\bar \phi \equiv 1, \bar W \equiv0)$ solves the system \eqref{MB} with the data $\bar \tau$ and $\bar \sigma$, one has
\[ -\frac{1}{8}R +\frac{1}{8}|\bar \sigma|^2-\frac{1}{12}\bar\tau^2=0. \]
Hence,
\[ \Delta-\frac{1}{8}R-\frac{7}{8}|\bar\sigma|^2-\frac{5}{12}\bar\tau^2= \Delta -|\bar\sigma|^2-\frac{1}{3}\bar\tau^2=\Delta-|K|^2. \]
Clearly, it is a negatively definite operator.

Finally, the theorem follows from the implicit function theorem.
\end{proof}

\begin{rem} For $K\equiv 0$, one can set $W\equiv0$, and the system \eqref{MB} reduces the Yamabe problem \cite{trudinger,Aubin,schoen84}.
\end{rem}

Note that the scaling symmetry discussed in Remark \ref{scaling} can be used to produce new solutions from the already obtained ones. In particular, one can obtain solutions with $\tau$ deviating from the vicinity of $\bar \tau$, at a cost of rescaling $\bar \sigma$. When the seed solution $(M, g, K)$ is a maximal slice, one can also produce new non-CMC initial data with the following scaling argument.

\begin{thm}\label{exi}
Suppose that we already have vacuum initial data $(M,g,K)$ with $\tr_gK=0$. Suppose $K\neq 0$ for some region. Assume further that $(M,g)$ has no conformal Killing vector fields. Given any $\tau \in W^{1,p}$, there is a positive constant $\eta>0$ such that for any $ \mu\in(0,\eta)$, there exists at least one solution $(\phi,W) \in W^{2,p}_+\times W^{2,p}$ of system \eqref{MB} for the data $(\hat\si=\mu^{12}K, \hat\tau=\mu^{-1}\tau)$.
\end{thm}

\begin{proof}
Since $(M,g,K)$ constitute vacuum initial data, system \eqref{MB} admits a particular solution $(\bar\phi \equiv 1, \bar W \equiv 0)$ for $\bar \tau=0$ and $\bar \si=K$.

Let us consider the following $\mu$-deformed system corresponding to Eqs.\ \eqref{MB}:
\begin{align*}
\MoveEqLeft \mathcal{G} \colon \mathbb{R} \times W^{2,p}_+\times W^{2,p} \rightarrow L^p\times L^p, \\
& \begin{pmatrix} \mu \\ \phi\\ W \end{pmatrix} \mapsto \begin{pmatrix}
\Delta \phi -\frac{1}{8}R \phi + \frac{1}{8}|K|^2\phi^{-7}+\frac{1}{4}\mu^4 \langle K, LW \rangle \phi^{-1}- \left( \mu^{10}\frac{1}{12}\tau^2-\mu^{8}\frac{1}{8}|LW|^2 \right) \phi^{5} \\
\nabla_i(LW)^i_j-\frac{2}{3}
\mu\nabla_j\tau+6(LW)^i_j\nabla_i\log \phi
\end{pmatrix}.
\end{align*}
It is easy to see that $\mathcal{G}$ is a $C^1$-mapping. The condition that $(M,g,K)$ constitute vacuum initial data with $\tr_g K = 0$ implies that $\mathcal{G}(0,1,0)=(0,0)$. We now prove that the partial derivative of $\mathcal{G}$ with respect to the variables $(\phi,W)$ is an isomorphism at $(0,\bar \phi\equiv1,\bar W\equiv0)$. The differential at the point $(0, \bar \phi \equiv 1,\bar W \equiv 0)$ is given by
\begin{align*}
\MoveEqLeft D\mathcal{G}|_{(0,1,0)} \begin{pmatrix} \delta\phi \\ \delta W \end{pmatrix} \\
&= \begin{pmatrix}
\Delta-\frac{1}{8}R-\frac{7}{8}|K|^2 & , & 0\\
 0 & , & \Delta_L
\end{pmatrix}
\begin{pmatrix}\delta\phi \\ \delta W \end{pmatrix},
\end{align*}
where $\Delta_L W=\mbox{div}_g(LW)$. Since $(0, \bar \phi \equiv 1,\bar W \equiv0)$ solves system \eqref{MB}, one has
\[ \Delta-\frac{1}{8}R-\frac{7}{8}|K|^2= \Delta -|K|^2. \]
The invertibility of the derivative $D \mathcal G|_{(0,1,0)}$ follows from the arguments stated already in the proof of Theorem \ref{3}.

By the implicit function theorem, for a sufficiently small parameter $\mu$, there exists $(\phi_{\mu}, W_\mu)$ such that $\mathcal{G}(\mu,\phi_\mu,W_\mu)=0$.

Define $\hat \phi_\mu=\mu^3 \phi_\mu$ and $\hat W_\mu=\mu^{-2}W_\mu$. Direct calculations show that $(\hat\phi_\mu, \hat W_\mu)$ solves system \eqref{MB} for the rescaled data $(\hat\si=\mu^{12}K, \hat\tau=\mu^{-1}\tau)$.
\end{proof}

In order to convert the above theorems into existence results, one has to assert the existence of solutions with $\tau = \bar\tau = \mathrm{const}$. Fortunately, in this case $W \equiv 0$ is a solution of \eqref{momentum2}. We only need to consider equation \eqref{hamiltonian2} with $W\equiv0$:
\begin{equation}
\label{newlich}
\Delta \phi -\frac{1}{8}R \phi + \frac{1}{8}|\bar\si|^2\phi^{-7}-\frac{1}{12}\bar\tau^2\phi^{5} = 0.
\end{equation}

As a consequence of Theorem 4.10 in \cite{Ch04}, one has
\begin{thm}\label{cmc}
If $\bar\tau$ is a constant, there exists $(\phi_{\bar\tau},W_{\bar\tau}\equiv0)$ solving equation \eqref{newlich}, provided that one of the following conditions holds:
\begin{enumerate}
\item The Yamabe constant $Y(g) > 0$ and $|\bar\si|^2 > 0$
\item $\bar\tau\neq 0$, $Y(g)=0$ and $|\bar\si|^2 > 0$
\item $\bar\tau\neq 0$ and $Y(g) < 0$
\end{enumerate}
\end{thm}

Here the Yamabe constant is defined as
\[ Y(g) = \mathrm{inf}_{u \in C^\infty(M), u \not\equiv 0} \frac{\int_M \left( |\nabla u|^2 + \frac{1}{8} R u^2 \right)dv_g}{\left( \int_M u^6 dv_g \right)^\frac{1}{3}}. \]
The set $\left(M,\phi_{\bar\tau}^4g, \frac{\bar \tau}{3} \phi_{\bar \tau}^4 g + \phi_{\bar\tau}^{-2}\bar\sigma, \bar \tau \right)$ can be regarded as the assumed CMC vacuum initial data in Theorem \ref{3} and, if $\bar \tau \equiv 0$, also in Theorem \ref{exi}.


\section{Conformally covariant split system on a compact manifold with boundary}
\label{boundarycase}

We will now turn to a physically important case, where the manifold $M$ is compact and has a black hole type boundary. Since we are working at the level of initial data, basically the only available concept of a black hole is that of a region enclosed within an apparent horizon. Consequently, the boundary conditions that we now adopt guarantee that the boundary of $M$ is a marginally trapped surface (and that the manifold $M$ constitutes the `exterior' of the black hole). It is however not surprising that requiring the boundary of $M$ to correspond to an apparent horizon is not sufficient as a prescription of boundary conditions for system (\ref{MB}). Consequently, a part of the boundary conditions have to be imposed in a more or less arbitrary manner. In this work we follow a recipe proposed in \cite{GN14}.

Let $(M,\widetilde g)$ be a compact 3-dimensional manifold with the boundary $\partial M$, and let $\widetilde \nu$ be the unit vector normal to $\partial M$. We assume that $\widetilde \nu$ is pointing `outwards' of $M$, and therefore to the `inside' of the black hole. The two null expansions of $\partial M$ are given by
\[ \Theta_\pm = \mp H_{\widetilde g} - K(\widetilde \nu, \widetilde \nu) + \tr_{\widetilde g} K, \]
where $H_{\widetilde g} = \widetilde \nabla_i \widetilde \nu^i$, and $\widetilde \nabla_i$ denotes the covariant derivative with respect to metric $\tilde g$. The condition that $\partial M$ is a marginally trapped surface can be stated as $\Theta_+ = 0$, $\Theta_- \le 0$. Let us further observe that $\frac{1}{2}\left( \Theta_- + \Theta_+ \right) = - K(\widetilde \nu, \widetilde \nu)  + \tr_{\widetilde g} K$, and $\frac{1}{2}\left( \Theta_- - \Theta_+ \right) = H_{\widetilde g}$. Consequently, the condition $\Theta_+ = 0$ yields
\begin{align}
\frac{1}{2}\Theta_- &= - K(\widetilde \nu, \widetilde \nu) + \tr_{\widetilde g} K, \label{thetam} \\
\frac{1}{2}\Theta_- &= H_{\widetilde g}. \label{thetap}
\end{align}
If the metric $\widetilde g$ and the extrinsic curvature $K$ are obtained as $\widetilde g = \phi^4 g$ and $K = \frac{\tau}{3} \phi^{4} g + \phi^{-2} \si + \phi^{4} LW$, the above conditions can be also expressed in terms of quantities related directly to $(M,g)$. A vector $\nu$, normal to $\partial M$, and normalized with respect to metric $g$, is related with $\widetilde \nu$ by $\nu^i = \phi^2 \widetilde \nu^i$. Also
\[ H_{\widetilde g} = \phi^{-2} H_g + 4 \phi^{-3} \partial_\nu \phi, \]
where $H_g = \nabla_i \nu^i$ is the mean curvature of $\partial M$ with respect to the metric $g$. Similarly,
\[ K(\widetilde \nu, \widetilde \nu)  = \frac{1}{3} \tau + \phi^{-6} \sigma(\nu,\nu) + LW(\nu,\nu). \]
Conditions (\ref{thetam}) and (\ref{thetap}) can be now rewritten as
\begin{equation}
\label{prototype}
\phi^{-6} \sigma(\nu,\nu) + LW(\nu,\nu) - \frac{2}{3} \tau + \frac{1}{2}\Theta_- = 0
\end{equation}
and
\[ \partial_\nu \phi + \frac{1}{4}H_g \phi - \frac{\Theta_-}{8} \phi^3 = 0. \]
Since Eq.\ (\ref{prototype}) is not sufficient as a boundary condition for $W$, we will actually replace it with a stronger requirement. Let $\xi$ denote a 1-form tangent to the boundary $\partial M$. We will require, as a boundary condition, that
\begin{equation}
\label{boundaryw}
\phi^{-6} \sigma(\nu,\cdot)  + LW(\nu,\cdot) - \frac{2}{3} \tau \nu^\flat + \frac{1}{2}\Theta_- \nu^\flat  - \xi = 0,
\end{equation}
where $\nu^\flat$ is the 1-form dual to the normal vector field $\nu$. Clearly, Eq.\ (\ref{prototype}) follows from Eq.\ (\ref{boundaryw}), as $\xi(\nu) = 0$.

In the remaining part of this section, we always assume that
\begin{equation}
\label{boundarysigma}
\si(\nu,\cdot)=0
\end{equation}
on $\partial M$. It is, of course, a restriction on the allowed forms of $\sigma$. On the other hand, it facilitates the proofs of the theorems of this section. In addition, assuming (\ref{boundarysigma}) ensures the $L^2$-orthogonality of $LW$ and $\sigma$. This can be easily demonstrated by a direct computation
\begin{align*}
\int_M \langle LW, \sigma \rangle dv_g &= 2 \int_M \sigma^{ij} \nabla_i W_j dv_g \\
&= 2 \int_{\partial M} \sigma_{ij} \nu^i W^j d \sigma_g = 0,
\end{align*}
where the boundary term vanishes because of assumption (\ref{boundarysigma}). Note that on compact manifolds without boundary, the $L^2$-orthogonality of $LW$ and $\sigma$ is guaranteed without any additional conditions.

In summary, we are now dealing with the following set of equations
\begin{subequations}\label{BMB}
\begin{align}
\Delta \phi -\frac{1}{8}R \phi + \frac{1}{8}|\si|^2\phi^{-7}+\frac{1}{4} \langle \si, LW \rangle \phi^{-1}- \left( \frac{1}{12}\tau^2-\frac{1}{8}|LW|^2 \right)\phi^{5} &= 0, \label{Bhamiltonian} \\
\nabla_i(LW)^i_j-\frac{2}{3} \nabla_j\tau+6(LW)^i_j\nabla_i\log \phi &= 0, \label{Bmomentum} \\
\partial_{\nu}\phi +\frac{1}{4}H\phi -\frac{\Theta_{-}}{8}\phi^3 &= 0, \label{Boundary1} \\
LW(\nu,.)-\frac{2}{3}\tau\nu^\flat+\frac{\Theta_{-}}{2}\nu^\flat-\xi &= 0, \label{Boundary2}
\end{align}
\end{subequations}
where \eqref{Boundary1} and \eqref{Boundary2} are the boundary conditions on $\partial M$. Here $g\in W^{2,p}(M)$, $\sigma \in W^{1,p}(M)$, $\tau \in W^{1,p}(M)$, $\Theta_- \in W^{1-\frac{1}{p},p}(\partial M)$, $\Theta_- \le 0$, and $\xi\in W^{1-\frac{1}{p},p}( \partial M)$ are the assumed data. In Equations (\ref{BMB}), and in the remaining part of this section, we drop, for convenience, the subscript $g$ in the symbol denoting the mean curvature of $\partial M$; we write simply $H_g \equiv H$.

Suppose that we already have vacuum initial data $(M,g,K)$, and denote $\bar \tau = \tr_g K$. Assume further that $\bar \tau$ is a constant. In this case the traceless part of $K$, $\bar\si_{ij}=K_{ij}-\frac{\bar\tau}{3}g_{ij}$ is divergence free. Then the pair $(\bar \phi \equiv 1, \bar W \equiv 0)$ is a solution of Eqs.\ (\ref{Bhamiltonian}) and (\ref{Bmomentum}). We will now require $(\bar \phi \equiv 1, \bar W \equiv 0)$ to also solve Eqs.\ (\ref{Boundary1}) and (\ref{Boundary2}). This implies that $\xi \equiv 0$ and $\Theta_- = 2 H = \frac{4}{3} \bar \tau = \mathrm{const}$.

Using the implicit function theorem, we can now assert the existence of a family of solutions to Eqs.\ (\ref{BMB}).

\begin{thm} Let $(M,g,K)$ be vacuum initial data with boundary $\partial M$ such that $\bar \tau = \tr_g K = \frac{3}{2} H = \mathrm{const} \le 0$, where $H$ denotes the mean curvature of $\partial M$. Let $\Theta_- = \frac{4}{3} \bar \tau$ and $\xi \equiv 0$ so that Eqs.\ (\ref{BMB}) admit a solution $(\bar \phi \equiv 1, \bar W \equiv 0)$. Assume further that $(M,g)$ has no conformal Killing vector fields, and $K \neq 0$ in some region of M. There is a small neighborhood of $\bar \tau$ in $W^{1,p}(M)$ such that for any $\tau$ in this neighborhood there exists a solution $(\phi_\tau,W_\tau) \in W^{2,p}_+(M)\times W^{2,p}(M)$ of system \eqref{BMB}.
\end{thm}

\begin{proof}
First, let us define a mapping
\begin{align*}
\MoveEqLeft \mathcal{F} \colon W^{1,p}(M) \times W^{2,p}_+(M)\times W^{2,p}(M) \rightarrow L^p(M)\times W^{1-\frac{1}{p},p}( \partial M) \times L^p (M) \times W^{1-\frac{1}{p},p}( \partial M), \\
& \begin{pmatrix} \tau\\ \phi\\ W \end{pmatrix} \mapsto \begin{pmatrix}
\Delta \phi -\frac{1}{8}R \phi + \frac{1}{8}|\bar \sigma|^2\phi^{-7}+\frac{1}{4} \langle \bar \sigma, LW \rangle \phi^{-1}- \left( \frac{1}{12}\tau^2-\frac{1}{8}| LW|^2 \right) \phi^{5} \\
\partial_{\nu}\phi +\frac{1}{4} H\phi -\frac{\Theta_{-}}{8}\phi^3 \\
\nabla_i(LW)^i_j-\frac{2}{3}\nabla_j\tau+6(LW)^i_j\nabla_i\log \phi \\
LW(\nu,.)-(\frac{2}{3}\tau-\frac{\Theta_{-}}{2})\nu^\flat
\end{pmatrix}.
\end{align*}
It is easy to see that $\mathcal{F}$ is a $C^1$-mapping and $\mathcal{F}(\bar{\tau},1,0)=(0,0,0,0)$. We prove that the partial derivative of $\mathcal{F}$ with respect to variables $(\phi,W)$ is an isomorphism at $(\bar{\tau},1,0)$. The differential at the point $(\bar{\tau},1,0)$ is given by
\begin{align*}
\MoveEqLeft D\mathcal{F}|_{(\bar{\tau},1,0)} \begin{pmatrix} \delta\phi \\ \delta W \end{pmatrix}\\
&= \begin{pmatrix}
\Delta-\frac{1}{8}R-\frac{7}{8}|\bar\sigma|^2-\frac{5}{12}\bar\tau^2 & , & \frac{1}{4} \langle \bar\sigma, L(\cdot) \rangle \\
\partial_{\nu}+\frac{1}{4} H-\frac{3\Theta_{-}}{8}  & , & 0\\
 0 & , & \Delta_L\\
 0 & , &L\cdot(\nu,\cdot)
\end{pmatrix}
\begin{pmatrix}\delta\phi \\ \delta W\end{pmatrix},
\end{align*}
and it is block triangular. Thus, the invertibility of $D\mathcal{F}|_{(\bar\tau,1,0)}$ follows from the fact that both diagonal block terms are invertible. More precisely:

\textit{Claim 1.} The map
\begin{align*}
\MoveEqLeft \mathcal{H}\colon W^{2,p}(M) \rightarrow L^p (M)\times W^{1-\frac{1}{p},p}( \partial M), \\
& \delta\phi \mapsto \left( \left(\Delta-\frac{1}{8}R-\frac{7}{8}|\bar\sigma|^2-\frac{5}{12}\bar\tau^2 \right)\delta \phi, \left(\partial_{\nu} +\frac{1}{4}H - \frac{3\Theta_-}{8}\right)\delta \phi \right)
\end{align*}
is invertible.

\textit{Claim 2.} The map
\begin{align*}
\MoveEqLeft \mathcal{L} \colon W^{2,p}(M) \rightarrow L^p (M) \times  W^{1-\frac{1}{p},p}( \partial M), \\
& \delta W \mapsto \left(\Delta_L \delta W,L\delta W(\nu,\cdot) \right)
\end{align*}
is also invertible.

Proof of \textit{Claim 1.} Note that $\mathcal{H}$ is a Fredholm operator of index $0$. It suffices to show that $\mathcal{H}$ is injective. Since the pair $(\bar \phi \equiv 1, \bar W \equiv 0)$ solves system \eqref{MB} for the mean curvature $\bar \tau$, one has
\[ -\frac{1}{8}R +\frac{1}{8}|\bar \sigma|^2-\frac{1}{12}\bar\tau^2=0, \]
and hence
\[ \Delta-\frac{1}{8}R-\frac{7}{8}|\bar\sigma|^2-\frac{5}{12}\bar\tau^2= \Delta -|\bar\sigma|^2-\frac{1}{3}\bar\tau^2 = \Delta - |K|^2. \]
Note that for any $\delta \phi \in W^{2,p}$ we have
\[ \int_M (\delta \phi) \Delta (\delta \phi) dv_g = -\int_M |\nabla (\delta \phi)|^2 dv_g + \int_{\partial M} (\delta \phi) \partial_{\nu} (\delta \phi) d\sigma_g. \]
If $\delta \phi \in \mathrm{ker} \mathcal{H}$, then one obtains
\[  \int_M |K|^2 (\delta \phi)^2 dv_g = - \int_M |\nabla (\delta \phi)|^2 dv_g + \int_{\partial M} \left(\frac{3\Theta_{-}}{8}-\frac{1}{4} H\right) (\delta \phi)^2 d \sigma_g. \]
Since $\frac{1}{4} H = \frac{\Theta_{-}}{8}$, we get
\[ \int_M |K|^2 (\delta \phi)^2 dv_g = - \int_M |\nabla (\delta \phi)|^2 dv_g + \int_{\partial M} \frac{1}{4} \Theta_- (\delta \phi)^2 d \sigma_g \leq 0, \]
where the last inequality follows from the assumption that $\Theta_{-} \leq 0$. The condition $K \neq 0$ implies that $|K|^2 > 0$, and hence $\delta \phi \equiv 0$. Therefore $\mathcal{H}$ has a trivial kernel, and it is invertible.

The proof of \textit{Claim 2} is a consequence of the assumption that $(M,g)$ admits no conformal Killing vector fields. The existence of solutions $(\phi_\tau, W_\tau)$ for $\tau$ in a neighborhood of $\bar \tau$ follows now from the implicit function theorem.
\end{proof}

\begin{thm}\label{3-mu}
Suppose that $(M,g,K)$ satisfy the vacuum Einstein's constraint equations, and $M$ has a boundary $\partial M$ such that $H \equiv 0$ on $\partial M$. Assume that $\tr_g K = 0$ and $K \neq 0$ in some region of $M$. Assume further that $(M,g)$ has no conformal Killing vector fields. Given any data $\tau \in W^{1,p}(M)$, $\Theta_- \in W^{1-\frac{1}{p},p}(\partial M)$, $\Theta_- \le 0$, and $\xi\in W^{1-\frac{1}{p},p}( \partial M)$, there is a positive constant $\eta>0$ such that for any $\mu \in (0, \eta)$, there exists at least one solution $(\phi, W) \in W^{2,p}_+(M)\times W^{2,p}(M)$ of the system \eqref{BMB} for the data $(\hat\sigma = \mu^{12} K, \hat\tau = \mu^{-1} \tau, \Theta_-, \xi)$.
\end{thm}

\begin{proof}
Let us start with defining a map
\begin{align*}
\MoveEqLeft \mathcal{F} \colon \mathbb{R} \times W^{2,p}_+(M)\times W^{2,p}(M) \rightarrow L^p(M)\times W^{1-\frac{1}{p},p}( \partial M) \times L^p (M) \times W^{1-\frac{1}{p},p}( \partial M),\\
& \begin{pmatrix} \mu\\ \phi\\ W \end{pmatrix} \mapsto \begin{pmatrix}
\Delta \phi -\frac{1}{8}R \phi + \frac{1}{8}|K|^2 \phi^{-7} + \frac{1}{4} \mu^4 \langle K, LW \rangle \phi^{-1} - \left( \frac{\mu^{10}}{12}\tau^2 - \frac{\mu^{8}}{8} |LW|^2 \right) \phi^{5}  \\
\partial_\nu \phi - \frac{\mu^6}{8} \Theta_{-} \phi^3 \\
\nabla_i(LW)^i_j - \frac{2}{3} \mu \nabla_j \tau + 6(LW)^i_j \nabla_i \log \phi \\
LW(\nu,.) - \left( \frac{2}{3} \mu \tau - \frac{\mu^2 \Theta_{-}}{2} \right) \nu^\flat - \mu^2 \xi
\end{pmatrix}.
\end{align*}
It is easy to see that $\mathcal{F}$ is a $C^1$-mapping and $\mathcal{F}(0,1,0) = (0,0,0,0)$. We prove that the partial derivative of $\mathcal{F}$ with respect to the variables $(\phi,W)$ is an isomorphism at $(0,1,0)$. The differential at the point $(0,1,0)$ is given by
\begin{align*}
\MoveEqLeft D\mathcal{F}|_{(0,1,0)} \begin{pmatrix} \delta\phi \\ \delta W \end{pmatrix}\\
&=\begin{pmatrix}
\Delta-\frac{1}{8}R-\frac{7}{8}|K|^2 & , & 0 \\
\partial_{\nu} & , & 0\\
 0 & , & \Delta_L\\
 0 & , &L\cdot(\nu,\cdot)
\end{pmatrix}
\begin{pmatrix}\delta\phi \\ \delta W\end{pmatrix}.
\end{align*}
Note that since $\mathcal{F}(0,1,0) = (0,0,0,0)$, we have $R = |K|^2$, and thus $\Delta-\frac{1}{8}R-\frac{7}{8}|K|^2 = \Delta - |K|^2$. One can easily check that $D\mathcal{F}|_{(0,1,0)}$ is invertible.

The implicit function theorem ensures now that for a sufficiently small parameter $\mu$ there exists a pair $(\phi_{\mu},W_\mu)$ such that $\mathcal{G}(\mu,\phi_\mu,W_\mu)=0$.

Define $\hat \phi_\mu=\mu^3 \phi_\mu$ and $\hat W_\mu=\mu^{-2}W_\mu$ by rescaling. Direct calculations show that $(\hat\phi_\mu, \hat W_\mu)$ solves the system \eqref{BMB} for the data $(\hat\si=\mu^{12} K, \hat\tau = \mu^{-1} \tau,  \Theta_-, \xi)$.
\end{proof}

There exist initial data satisfying the assumptions of Theorem \ref{3-mu}. From the work of Escobar \cite{E}, we know that there exists a conformal factor $\psi>0$ such that $(M,\hat g=\psi^4 g)$ has a constant scalar curvature $R(\hat g)$ and the mean curvature $H_{\hat g}$ of $\partial M$ vanishes. Thus, we can further safely consider a seed manifold $(M,g)$ with $R(g)= \mathrm{const}$ and $H \equiv 0$ on the boundary. Let $\bar \sigma$ be a symmetric TT-tensor of type $(0,2)$. We take particular data $\bar \tau\equiv0$, $\bar\Theta_-\equiv0$, $\bar \xi\equiv0$. It is clear that $W_{\bar \tau}\equiv0$ solves \eqref{Bmomentum} and \eqref{Boundary2}.

In the presence of a boundary, the Yamabe constant is defined as
\[ Y(g,\partial M)= \mathrm{inf}_{u \in C^{\infty}(M), u \not\equiv 0} \frac{\int_M \left( |\nabla u|^2+\frac{1}{8}Ru^2 \right)dv_g + \int_{\partial M} \frac{1}{4} H u^2 d\sigma_g}{\left( \int_M u^6 dv_g \right)^{\frac{1}{3}}}. \]
If $Y(g,\partial M)>0$, $|\bar \sigma|^2>0$, and $(M,g)$ has no conformal Killing vectors, then an argument similar to that of Theorem 4.10 of \cite{Ch04} shows that there is a sub-solution and a super-solution of the equation
\begin{equation}
\Delta \phi -\frac{1}{8}R \phi + \frac{1}{8}|\bar\si|^2\phi^{-7} = 0\end{equation}
with the Neumann boundary condition $\partial_\nu \phi=0$. Therefore there exists a pair $(\phi_{\bar\tau}, W_{\bar \tau}\equiv0)$ solving system \eqref{BMB} with data $\bar \tau\equiv0$, $\bar\Theta_-\equiv0$, $\bar \xi\equiv0$ and the set $(M,\phi_{\bar\tau}^4g, \phi_{\bar\tau}^{-2}\bar\sigma)$ can be regarded as the assumed CMC vacuum initial data in Theorem \ref{3-mu}.


\section{Conformally covariant split system with the cosmological constant}
\label{withlambda}

Many of the results of preceding sections can be easily generalized to the case with a non-zero cosmological constant $\Lambda$. In what follows we discuss these generalizations. In most cases we will omit the proofs, as they are analogous to those given in Sec.\ \ref{closed}. We will, however, focus on the places, where taking into account a non-zero cosmological constant introduces a change with respect to $\Lambda = 0$ cases, and where some additional assumptions are required.

The Einstein vacuum constraint equations with a cosmological constant $\Lambda$ read
\begin{subequations}\label{constraintslambda}
\begin{align}
R_{\widetilde g} - |K|^2_{\widetilde g} + (\tr_{\widetilde g} K)^2 - 2 \Lambda &= 0, \\
\divergence_{\widetilde g} K - d \tr_{\widetilde g} K &= 0.
\end{align}
\end{subequations}
Keeping standard definitions, i.e., $LW$ defined by Eq.\ (\ref{LWdef}) and $K = \frac{\tau}{3}\phi^4 g + \phi^{-2} \sigma + \phi^4 LW$, we get the system
\begin{subequations}\label{MBlambda}
\begin{align}
\Delta \phi -\frac{1}{8}R \phi + \frac{1}{8}|\si|^2\phi^{-7}+\frac{1}{4} \langle \si, LW \rangle \phi^{-1} & \nonumber \\
- \left( \frac{1}{12}\tau^2-\frac{1}{8}| LW|^2 - \frac{1}{4} \Lambda \right)\phi^{5} &= 0,\label{hamiltonian2lambda} \\
\nabla_i(LW)^i_j-\frac{2}{3}\nabla_j\tau+6(LW)^i_j\nabla_i\log \phi &= 0. \label{momentum2lambda}
\end{align}
\end{subequations}

\begin{rem}\label{scalinglambda}
Similarly to the scaling symmetry described in Remark \ref{scaling}, system (\ref{MBlambda}) admits the following scaling. Suppose that system (\ref{MBlambda}) has a solution $(\phi, W)$. Set $\hat \phi = \mu^{-\frac{1}{4}} \phi$, $\hat W = \mu^{\frac{1}{2}} W$ for some positive number $\mu \in \mathbb R^+$. Then $(\hat \phi, \hat W)$ satisfy system (\ref{MBlambda}) with the data $\hat \sigma$, $\hat \tau$, and the cosmological constant $\hat \Lambda$ given by $\hat \sigma = \mu^{-1} \sigma$, $\hat \tau = \mu^{\frac{1}{2}} \tau$, $\hat \Lambda = \mu \Lambda$.
\end{rem}

As before, we assume that $(M,g,K)$ already satisfies Eqs.\ (\ref{constraintslambda}). Denote $\bar \tau = \mathrm{tr}_g K$, and $\bar \sigma_{ij} = K_{ij} - \frac{1}{3} \bar \tau g_{ij}$. The following result is a straightforward generalization of Theorem \ref{3}.

\begin{thm}\label{lambda1}
Suppose that we already have vacuum initial data $(M,g,K)$ satisfying Eqs.\ (\ref{constraintslambda}) with $\bar \tau = \mathrm{tr}_g K = \mathrm{const}$. Assume that $(M,g)$ has no conformal Killing vector fields. Assume further that $-|K|^2 + \Lambda \le 0$ on $M$, and $-|K|^2 + \Lambda < 0$ in some region of $M$. There is a small neighborhood of $\bar \tau$ in $W^{1,p}$ such that for any $\tau$ in this neighborhood there exists $(\phi_\tau,W_\tau) \in W^{2,p}_+ \times W^{2,p}$ solving system \eqref{MBlambda} for the data $\bar\si_{ij}=K_{ij}-\frac{\bar \tau}{3}g_{ij}$ and $\tau$.
\end{thm}

Note that $-|K|^2 + \Lambda = - |\bar \sigma|^2 - \frac{1}{3} \bar \tau^2 + \Lambda$. In order to satisfy the condition $-|K|^2 + \Lambda < 0$ it is enough to require that $- \frac{1}{3} \bar \tau^2 + \Lambda < 0$.

\begin{proof}
The only difference with respect to the proof of Theorem \ref{3} is that now the function $\mathcal H \colon W^{2,p} \to L^p$ is defined by
\[ \delta \phi \mapsto \left[ \Delta - \frac{1}{8} R - \frac{7}{8} |\bar \sigma|^2 - \frac{5}{4} \left( \frac{\bar \tau^2}{3} - \Lambda \right)  \right] \delta \phi. \]
For $\phi \equiv 1$, $W \equiv 0$, we have from Eq.\ (\ref{hamiltonian2lambda})
\[ -\frac{1}{8}R + \frac{1}{8} |\bar \sigma|^2 - \frac{1}{4}\left( \frac{\bar \tau^2}{3} - \Lambda \right) = 0.  \]
Thus
\[ \Delta - \frac{1}{8} R - \frac{7}{8} |\bar \sigma|^2 - \frac{5}{4} \left( \frac{\bar \tau^2}{3} - \Lambda \right) = \Delta - |\bar \sigma|^2 - \frac{1}{3} \bar \tau^2 + \Lambda, \]
which is a negatively definite map, provided that
\[ - |\bar \sigma|^2 - \frac{1}{3} \bar \tau^2 + \Lambda \le 0 \]
on $M$, and
\[ - |\bar \sigma|^2 - \frac{1}{3} \bar \tau^2 + \Lambda < 0 \]
in some region of $M$. Noting that $- |\bar \sigma|^2 - \frac{1}{3} \bar \tau^2 = - |K|^2$, we conclude that $\mathcal H$ has a trivial kernel.
\end{proof}

Theorem \ref{exi} can be generalized in two ways. One can generate data corresponding to $\Lambda \neq 0$ from a seed-initial data with $\mathrm{tr}_g K = 0$ and $\Lambda = 0$, i.e., initial data satisfying Eqs.\ (\ref{originalconstraints}). The other possibility is to start with seed-initial data that already satisfy constraints (\ref{constraintslambda}) with $\mathrm{tr}_g K = 0$ and some nonzero value of $\Lambda$. These data can be then used to generate another set of initial data corresponding to some mean curvature $\hat \tau \neq 0$ and a different value of the cosmological constant $\hat \Lambda$. We will discuss these two cases in what follows.

\begin{thm}\label{exi_lambda_a}
Suppose that $(M, g, K)$ satisfy the constraint equations (\ref{originalconstraints}), and $\mathrm{tr}_g K = 0$. Let $K \neq 0$ in some region, and let $(M,g)$ admit no conformal Killing vector fields. Given any $\tau \in W^{1,p}$ and $\Lambda$, there is a positive constant $\eta > 0$ such that for any $\mu \in (0,\eta)$ there exists a solution $(\phi, W) \in W^{2,p}_+ \times W^{2,p}$ of system (\ref{MBlambda}) for the data $\hat \sigma = \mu^{12} K$, $\hat \tau = \mu^{-1} \tau$, and the cosmological constant $\hat \Lambda = \mu^{-2} \Lambda$.
\end{thm}

\begin{proof}
The function $\mathcal{G}$ appearing in the proof of Theorem \ref{exi} is now defined as
\begin{align*}
\MoveEqLeft \mathcal{G} \colon \mathbb{R} \times W^{2,p}_+ \times W^{2,p} \rightarrow L^p\times L^p, \\
& \begin{pmatrix} \mu \\ \phi\\ W \end{pmatrix} \mapsto \begin{pmatrix}
\Delta \phi -\frac{1}{8}R \phi + \frac{1}{8}|K|^2\phi^{-7}+\frac{1}{4}\mu^4 \langle K, LW \rangle \phi^{-1} - \left( \frac{\mu^{10} \tau^2}{12} - \frac{\mu^8}{8}| LW|^2 - \frac{\mu^{10}}{4} \Lambda \right) \phi^{5} \\
\nabla_i(LW)^i_j-\frac{2}{3}
\mu\nabla_j\tau+6(LW)^i_j\nabla_i\log \phi
\end{pmatrix}.
\end{align*}
The assumptions of the theorem imply that $\mathcal G(0,1,0) = (0,0)$. The differential of $\mathcal G$ with respect to $(\phi, W)$ at the point $(0, \bar \phi \equiv 1, \bar W \equiv 0)$ reads
\[ D\mathcal{G}|_{(0, 1, 0)} \begin{pmatrix} \delta\phi \\ \delta W \end{pmatrix}
 = \begin{pmatrix}
\Delta-\frac{1}{8}R-\frac{7}{8}|K|^2 & , & 0\\
 0 & , & \Delta_L
\end{pmatrix}
\begin{pmatrix}\delta\phi \\ \delta W\end{pmatrix}, \]
where $\Delta_L W=\mbox{div}_g(LW)$. There is no change with respect to the proof of Theorem \ref{exi} at this point. We conclude that for sufficiently small $\mu$, there exist $\phi_\mu$ and $W_\mu$ solving $\mathcal G(\mu, \phi_\mu, W_\mu) = 0$.

The scaling argument works analogously to the one used in Theorem \ref{exi}. We define $\hat \phi_\mu = \mu^3 \phi_\mu$, $\hat W_\mu = \mu^{-2} W_\mu$. Then the pair $(\hat \phi_\mu, \hat W_\mu)$ solves system (\ref{MBlambda}) for the data $(\hat \sigma = \mu^{12} K, \hat \tau = \mu^{-1} \tau)$, and the cosmological constant $\hat \Lambda = \mu^{-2} \Lambda$.
\end{proof}

\begin{thm}\label{exi_lambda_b}
Suppose that $(M, g, K)$ satisfy the constraint equations (\ref{constraintslambda}) with a non-zero cosmological constant $\Lambda$, and $\mathrm{tr}_g K = 0$. Assume that $(M,g)$ admit no conformal Killing vector fields. Assume further that $- |K|^2 + \Lambda \le 0$ on $M$ and $-|K|^2 + \Lambda < 0$ in some region of $M$. Given any $\tau \in W^{1,p}$, there is a positive constant $\eta > 0$ such that for any $\mu \in (0,\eta)$ there exists a solution $(\phi, W) \in W^{2,p}_+ \times W^{2,p}$ of system (\ref{MBlambda}) for the data $\hat \sigma = \mu^{12} K$, $\hat \tau = \mu^{-1} \tau$, and the cosmological constant $\hat \Lambda = \mu^{-12} \Lambda$.
\end{thm}

\begin{proof}
The proof is essentially a variant of the proof of Theorem \ref{exi_lambda_a}. The map $\mathcal G$ is defined as
\begin{align*}
\MoveEqLeft \mathcal{G} \colon \mathbb{R} \times W^{2,p}_+ \times W^{2,p} \rightarrow L^p\times L^p, \\
& \begin{pmatrix} \mu \\ \phi\\ W \end{pmatrix} \mapsto \begin{pmatrix}
\Delta \phi -\frac{1}{8}R \phi + \frac{1}{8}|K|^2\phi^{-7}+\frac{1}{4}\mu^4 \langle K, LW \rangle \phi^{-1} - \left( \frac{\mu^{10} \tau^2}{12} - \frac{\mu^8}{8}| LW|^2 - \frac{1}{4} \Lambda \right) \phi^{5} \\
\nabla_i(LW)^i_j-\frac{2}{3}
\mu\nabla_j \tau+6(LW)^i_j\nabla_i\log \phi
\end{pmatrix},
\end{align*}
i.e., with no $\mu$-dependent coefficient in front of the $\Lambda$ term. Injectivity of $D\mathcal G$ (with respect to $(\phi, W)$) at the point $(\mu = 0, \bar \phi \equiv 1, \bar W \equiv 0)$ is ensured by requiring that $-|K|^2 + \Lambda \le 0$ on $M$, and $-|K|^2 + \Lambda < 0$ in some region of $M$, together with the condition that $(M,g)$ admits no conformal Killing vector fields (cf. the proof of Theorem \ref{lambda1}). The implicit function theorem then yields that for a sufficiently small $\mu$, there exist $\phi_\mu$ and $W_\mu$ solving $\mathcal G(\mu, \phi_\mu, W_\mu) = 0$.

The scaling argument can be now used to infer the existence of solutions of system (\ref{MBlambda}), but at a cost of rescaling both $K$ and $\Lambda$. We define $\hat \phi_\mu = \mu^3 \phi_\mu$, $\hat W_\mu = \mu^{-2} W_\mu$. Then $\hat \phi_\mu$ and $\hat W_\mu$ solve system (\ref{MBlambda}) for the data $\hat \sigma = \mu^{12} K$, $\hat \tau = \mu^{-1} \tau$ and the cosmological constant $\hat \Lambda = \mu^{-12} \Lambda$.
\end{proof}

In the same spirit it is possible to generalize the results of Sec.\ \ref{boundarycase} referring to manifolds with black hole boundaries. We omit this discussion here; the corresponding theorems can be easily formulated and proved by combining the techniques used in this section and in Sec.\ \ref{closed}.


\section{Numerical examples}
\label{numerics}

In this section we give simple numerical examples of solutions of system (\ref{MB}) with non-constant $\tau$. Our examples are inspired by a recent numerical study of non-CMC solutions obtained in the framework of the standard conformal approach \cite{Dilts_numerical}. Numerical techniques used in this paper are essentially adapted from \cite{Mach_Knopik}. We should emphasize that the spectrum of solutions of systems introduced in this section could be complex and should probably be a subject of a separate investigation, similar to \cite{Dilts_numerical} or \cite{Mach_Knopik}. The main aim of providing our numerical examples in this work is to demonstrate the applicability of the conformally covariant split method in numerical studies.

We start by considering the manifold $M = \mathbb S^1 \times \mathbb S^2$ endowed with the metric
\[ g = d \alpha^2 + d\theta^2 + \sin^2 \theta d \varphi^2. \]
Here $\alpha$ is the coordinate on $\mathbb S^1$, and $(\theta,\varphi)$ are coordinates on $\mathbb S^2$. The scalar curvature associated with metric $g$ reads $R = 2$. We assume the tensor $\sigma$ in the form
\[ \sigma_{ij} = b \left( \begin{array}{ccc}
-2 & 0 & 0 \\
0 & 1 & 0 \\
0 & 0 & \sin^2 \theta
\end{array} \right) \]
[in coordinates $(\alpha, \theta, \varphi)$], where $b$ is a constant. An elementary calculation shows that $\sigma$ is both trace-less and divergence-free. We have $|\sigma|^2 = 6b^2$.

For the vector field $W^i$ we assume an ansatz $W^i = (f(\alpha),0,0)$. It follows that
\[ (LW)_{ij} = \frac{2}{3} f^\prime \left( \begin{array}{ccc}
2 & 0 & 0 \\
0 & -1 & 0 \\
0 & 0 & - \sin^2 \theta
\end{array} \right), \]
where the prime denotes the derivative with respect to $\alpha$. It is easy to show that $|LW|^2 = \frac{8}{3} (f^\prime)^2$ and $\langle \sigma, LW \rangle  = -4 b f^\prime$. We also have $\nabla_i (LW)^i_j = \left( \frac{4}{3} f^{\prime \prime}, 0, 0 \right)$.

We will assume further that the conformal factor $\phi$ depends only on $\alpha$. This gives
\[ (LW)^i_j \nabla_i \log \phi = \left( \frac{4}{3} f^\prime (\log \phi)^\prime, 0, 0 \right). \]
It follows from Eq.\ (\ref{momentum2}) that the above choices are only consistent with $\partial_\theta \tau = \partial_\varphi \tau = 0$. From now on we assume that the mean curvature $\tau$ depends only on $\alpha$. Equation (\ref{momentum2}) then yields
\begin{equation}
\label{momentumf}
f^{\prime \prime} - \frac{1}{2} \tau^\prime + 6 \phi^{-1} f^\prime \phi^\prime = 0.
\end{equation}

Equation (\ref{hamiltonian2}) can be written as
\begin{equation}
\label{hamiltonianf}
\phi^{\prime \prime} - \frac{1}{4} \phi + \frac{3}{4} b^2 \phi^{-7} - b f^\prime \phi^{-1} -  \frac{1}{3}\left( \frac{\tau^2}{4} - (f^\prime)^2 \right) \phi^5 = 0.
\end{equation}

Note that the function $f$ appears in Eqs.\ (\ref{momentumf}) and (\ref{hamiltonianf}) only through its derivatives. Consequently, we will now define $w = f^\prime$, and write Eqs.\ (\ref{momentumf}) and (\ref{hamiltonianf}) as
\begin{subequations}\label{systemS1S2}
\begin{align}
\phi^{\prime \prime} - \frac{1}{4} \phi + \frac{3}{4} b^2 \phi^{-7} - b w \phi^{-1} -  \frac{1}{3}\left( \frac{\tau^2}{4} - w^2 \right) \phi^5 &= 0, \label{systemS1S2hamiltonian} \\
w^\prime - \frac{1}{2} \tau^\prime + 6 \phi^{-1} w \phi^\prime &= 0. \label{systemS1S2momentum}
\end{align}
\end{subequations}
The above set of equations has to be solved assuming that $\phi$ and $w$ are both functions on $\mathbb S^1$. The condition that $w = f^\prime$, where $f$ is also a regular function on $\mathbb S^1$ yields, in particular, that
\begin{equation}
\label{whasnoconstantterm}
\int_0^{2\pi} w d\alpha = 0.
\end{equation}

System (\ref{systemS1S2}) has an elementary solution for $\tau = \mathrm{const}$. The condition $\tau^\prime = 0$ and Eq.\ (\ref{systemS1S2momentum}) give $w = C \phi^{-6}$, where $C$ is an integration constant. Since $\phi^{-6}$ is non-negative, it follows from Eq.\ (\ref{whasnoconstantterm}) that in fact $C = 0$ and $w \equiv 0$. Equation (\ref{systemS1S2hamiltonian}) is now reduced to
\begin{equation}
\label{systemS1S2zerow}
\phi^{\prime \prime} - \frac{1}{4} \phi + \frac{3}{4} b^2 \phi^{-7} -  \frac{1}{12} \tau^2 \phi^5 = 0.
\end{equation}
Let $b \neq 0$. It is easy to see that the above equation admits a positive solution $\phi = \phi_0 = \mathrm{const}$. Assuming $\phi = \phi_0$, we get from Eq.\ (\ref{systemS1S2zerow})
\[ \mathrm{LHS} \equiv -x^2 + 3 b^2 - \frac{1}{3} \tau^2 x^3 = 0, \]
where $x = \phi_0^4$. Clearly, for $x = 0$ we have $\mathrm{LHS} = 3b^2 \ge 0$. For $x \to \infty$ one gets $\mathrm{LHS} \to - \infty$. Thus there exist a positive solution $\phi = \phi_0$ for all constants $b \neq 0$ and $\tau$. It can be easily expressed by a Cardano formula. It can be also shown that a positive solution to Eq.\ (\ref{systemS1S2zerow}) (not necessarily constant) is unique \cite{Isenberg}. Thus the solution $\phi = \phi_0$ is the only positive one.

We now turn to the non-CMC case, and construct simple numerical solutions of system (\ref{systemS1S2}). We employ the following spectral numerical scheme. We express both functions $\phi$ and $w$ truncating a Fourier series:
\[ \phi = \frac{a_0}{2} + \sum_{k = 1}^N a_k \cos (k \alpha) + \sum_{k=1}^N b_k \sin(k \alpha), \]
\[ w = \sum_{k = 1}^N c_k \cos (k \alpha) + \sum_{k=1}^N d_k \sin(k \alpha). \]
Note that there is no constant term in the expression for $w$, due to condition (\ref{whasnoconstantterm}). These expansions are substituted into Eqs.\ (\ref{systemS1S2}). Equation (\ref{systemS1S2hamiltonian}) is then projected on functions $\cos(k \alpha)$, $k = 0, \dots, N$ and $\sin(k \alpha)$, $k = 1, \dots, N$. Equation (\ref{systemS1S2momentum}) is projected on  $\cos(k \alpha)$ and $\sin(k \alpha)$, $k = 1, \dots, N$. This yields $4N + 1$ nonlinear conditions on the coefficients $a_k$, $k = 0, \dots, N$ and $b_k$, $c_k$, $d_k$ with $k = 1, \dots N$. These conditions are then solved using a standard Newton-Raphson procedure.

The integrals
\[ \frac{1}{\pi} \int_0^{2 \pi} d \alpha f(\alpha) \cos (k \alpha), \quad \frac{1}{\pi} \int_0^{2 \pi} d \alpha f(\alpha) \sin (k \alpha) \]
are performed using standard formulas
\[ \frac{1}{\pi} \int_0^{2 \pi} d \alpha f(\alpha) \cos (k \alpha) = \frac{2}{\tilde N} \sum_{j = 1}^{\tilde N} f(\alpha_j) \cos (k \alpha_j), \]
\[ \frac{1}{\pi} \int_0^{2 \pi} d \alpha f(\alpha) \sin (k \alpha) = \frac{2}{\tilde N} \sum_{j = 1}^{\tilde N} f(\alpha_j) \sin (k \alpha_j), \]
where $\tilde N = 2N +1$, and $\alpha_j = 2 \pi (j -1)/\tilde N$, $j = 1, \dots, \tilde N$. Our numerical calculations were performed using Wolfram Mathematica \cite{mathematica}.

In the following examples we assume $\tau = 1 + a \cos(\alpha)$. This choice allows one to search for solutions for $\phi$ and $w$, for which $b_k = d_k = 0$, $k = 1, \dots, N$. Note that in such a case Eq.\ (\ref{systemS1S2momentum}) has to be projected on the set of functions $\sin(k \alpha)$, $k = 1, \dots, N$.

\begin{figure}[t]
\includegraphics[width=0.6\textwidth]{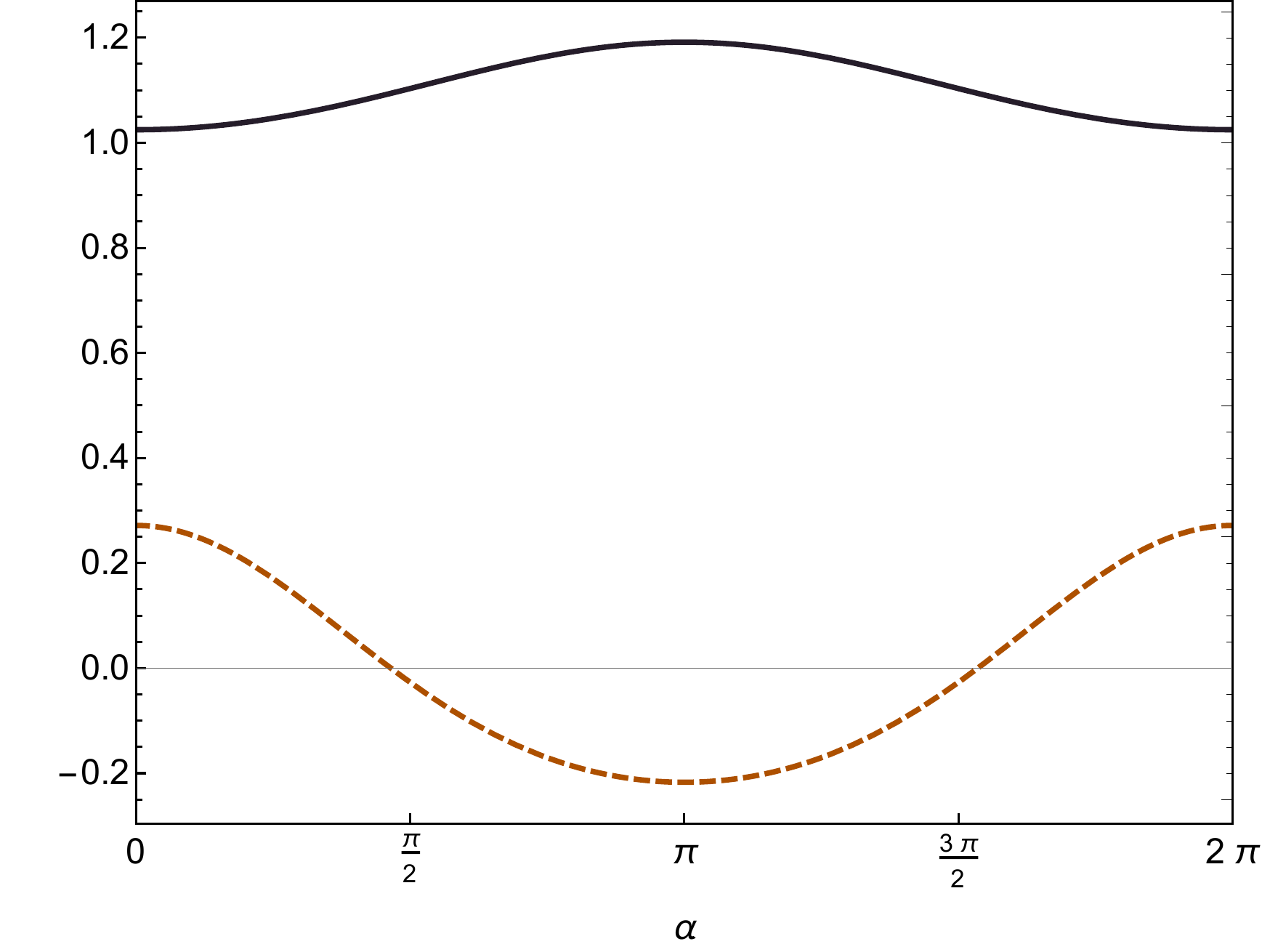}
\caption{\label{S1S2b1a12}A sample numerical solution of system (\ref{systemS1S2}) with $\tau = 1 + a \cos(\alpha)$, $a = 1/2$, $b = 1$. The solid line depicts the conformal factor $\phi$. The dashed line shows the graph of $w$.}
\end{figure}

\begin{figure}[t]
\includegraphics[width=0.6\textwidth]{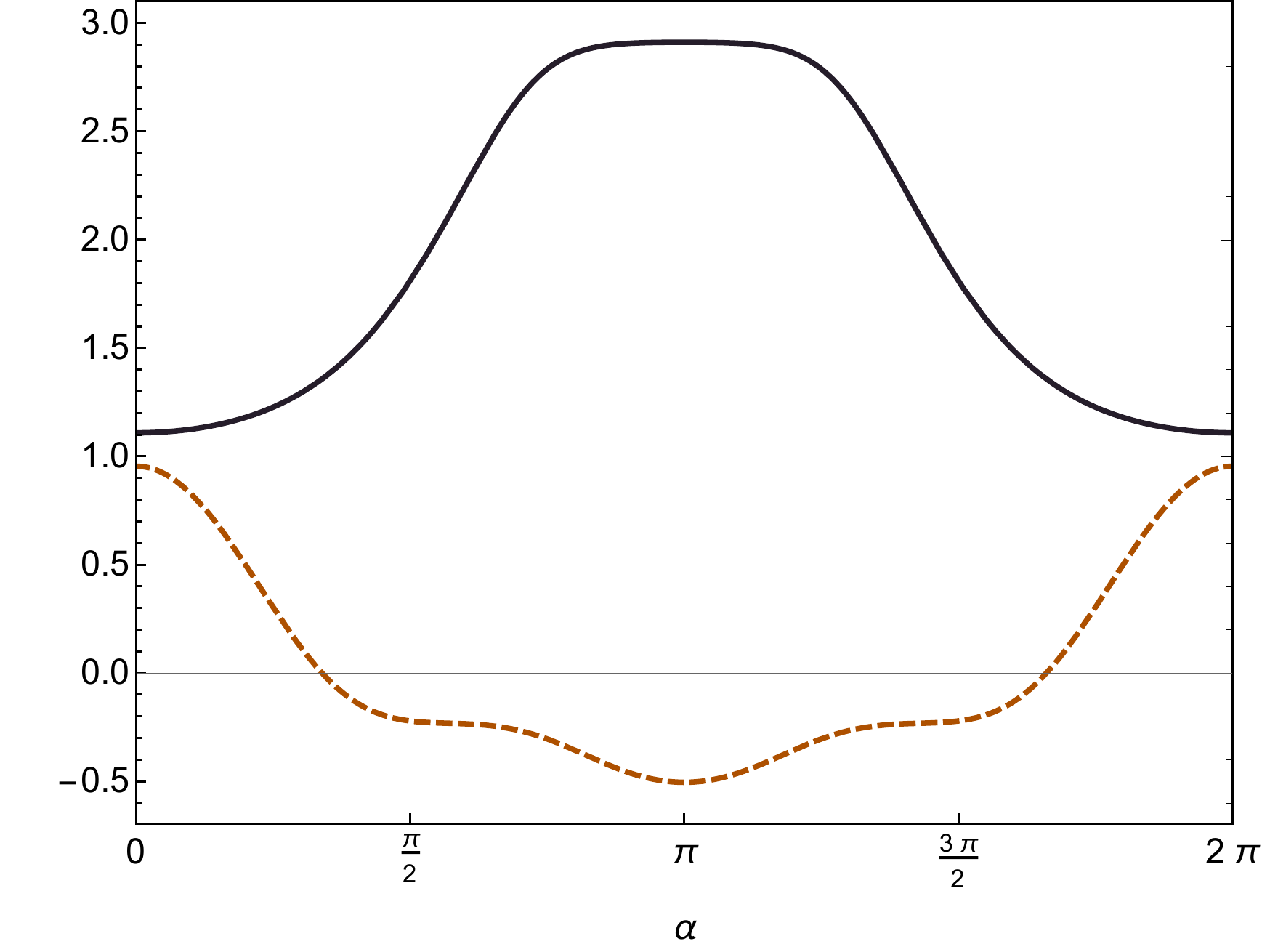}
\caption{\label{S1S2b3a2}Same as in Fig.\ \ref{S1S2b1a12}, but for $a = 2$, $b = 3$.}
\end{figure}

Sample solutions obtained for system (\ref{systemS1S2}) with $\tau = 1 + a \cos(\alpha)$ are shown in Figs.\ \ref{S1S2b1a12} and \ref{S1S2b3a2}. Figure \ref{S1S2b1a12} shows a solution corresponding to $a = 1/2$ and $b = 1$. Note that in this case $\tau$ is strictly positive. Figure \ref{S1S2b3a2} depicts a solution obtained for $a = 2$, $b = 3$. This means that the mean curvature $\tau$ changes its sign.

As usual for spectral methods, the minimum acceptable value of $N$ depends on the equation to be solved, and the desired accuracy. In the specific examples shown in Figs.\ \ref{S1S2b1a12} and \ref{S1S2b3a2} we chose $N = 25$ and $N = 55$, respectively. The quality of solutions can be assessed by computing the left-hand sides of Eqs.\ (\ref{systemS1S2}) for the obtained numerical solutions. For the solution shown in Fig.\ \ref{S1S2b1a12} these values are of the order of $10^{-15}$. For the solution depicted in Fig.\ \ref{S1S2b3a2} we get the left-hand sides of Eqs.\ (\ref{systemS1S2}) not exceeding a value of the order of $10^{-14}$.

\begin{figure}[th]
\includegraphics[width=0.6\textwidth]{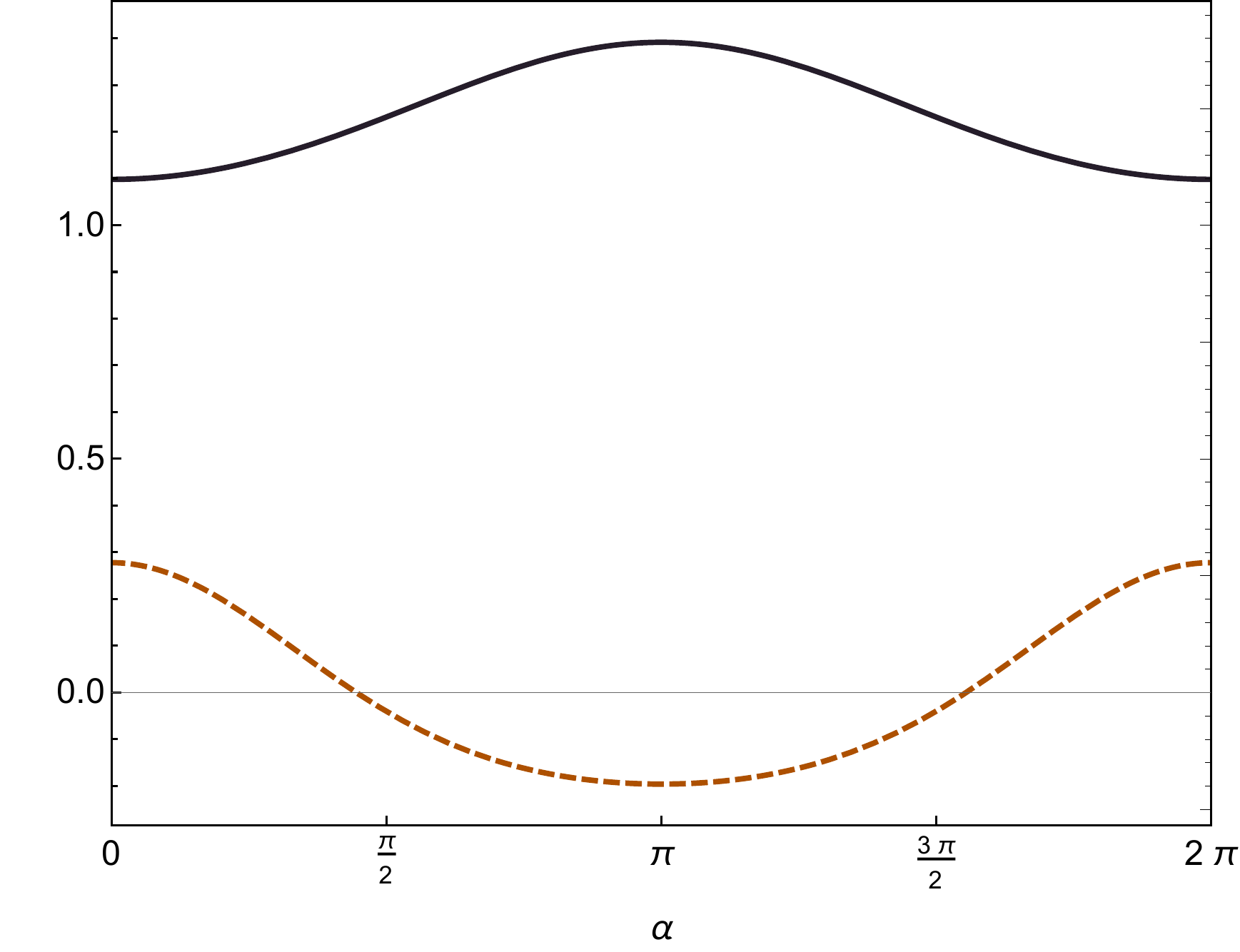}
\caption{\label{S1T2b1a12}A sample numerical solution of system (\ref{systemS1T2}) with $\tau = 1 + a \cos(\alpha)$, $a = 1/2$, $b = 1$. The solid line depicts the conformal factor $\phi$. The dashed line shows the graph of $w$.}
\end{figure}

\begin{figure}[th]
\includegraphics[width=0.6\textwidth]{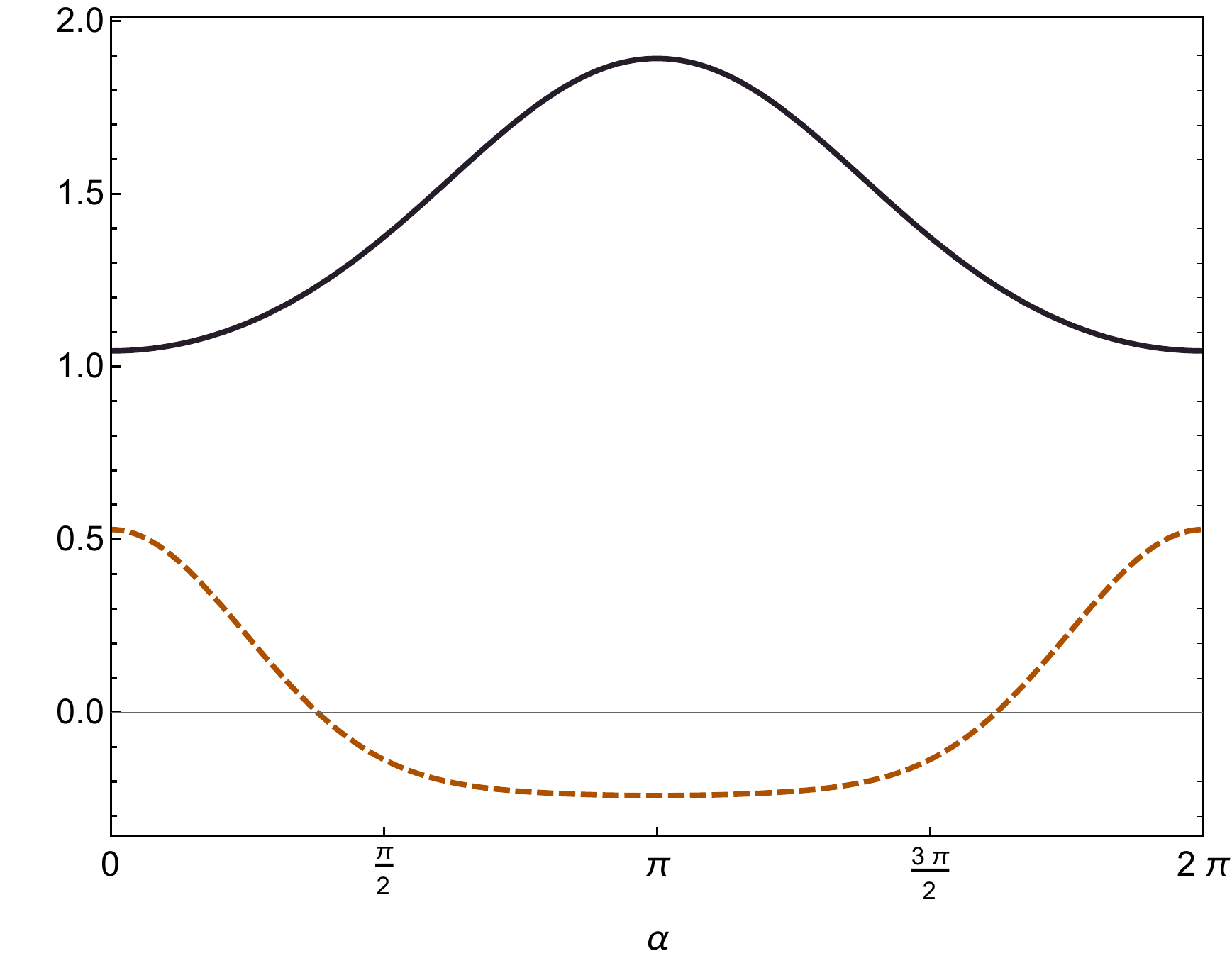}
\caption{\label{S1T2b1a1}Same as in Fig.\ \ref{S1T2b1a12}, but for $a = b = 1$.}
\end{figure}

There is another example that can be obtained in the similar fashion, and it is again motivated by a system analyzed in \cite{Dilts_numerical}. We take $M = \mathbb S^1 \times \mathbb T^2$, where the torus $\mathbb T^2 = \mathbb S^1 \times \mathbb S^1$. Let $(\alpha, \beta, \gamma)$ be the coordinates on $M$, each spanning the range $[0,2 \pi)$. We assume the metric $g$ to be flat, $g_{ij} = \mathrm{diag}(1,1,1)$, and $\sigma_{ij} = b \, \mathrm{diag}(-2, 1, 1)$, where $b$ is a constant. Clearly, $\sigma$ is a TT tensor. As before, $|\sigma|^2 = 6 b^2$. Let us also assume that $W^i = (f(\alpha),0,0)$. It follows that $(LW)_{ij} = \frac{2}{3}f^\prime \mathrm{diag}(2,-1,-1)$, $|LW|^2  = \frac{8}{3}(f^\prime)^2$, $\langle \sigma, LW \rangle  = - 4 b f^\prime$. Assuming that $\phi$ and $\tau$ depend only on $\alpha$, we also have $\nabla_i (LW)^i_j = \left(\frac{4}{3}f^{\prime \prime},0,0 \right)$, and
\[ (LW)^i_j \nabla_i \log \phi = \left( \frac{4}{3} f^\prime (\log \phi)^\prime, 0, 0\right). \]
It turns out that systems (\ref{MB}) written for the two cases $\mathbb S^1 \times \mathbb S^2$ and $\mathbb S^1 \times \mathbb T^2$ differ only in the term proportional to the scalar curvature $R$. We have $R = 0$ for the $\mathbb S^1 \times \mathbb T^2$ case. The analogue of Eqs.\ (\ref{systemS1S2}) can be now written as
\begin{subequations}\label{systemS1T2}
\begin{align}
\phi^{\prime \prime} + \frac{3}{4} b^2 \phi^{-7} - b w \phi^{-1} -  \frac{1}{3}\left( \frac{\tau^2}{4} - w^2 \right) \phi^5 &= 0, \label{systemS1T2hamiltonian} \\
w^\prime - \frac{1}{2} \tau^\prime + 6 \phi^{-1} w \phi^\prime &= 0, \label{systemS1T2momentum}
\end{align}
\end{subequations}
where again $w \equiv f^\prime$.

If $\tau^\prime = 0$, we again obtain a solution with $w = 0$. In this case Eq.\ (\ref{systemS1T2hamiltonian}) reads
\begin{equation}
\label{S1T2constt}
\phi^{\prime \prime} + \frac{3}{4} b^2 \phi^{-7} - \frac{\tau^2}{12} \phi^5 = 0.
\end{equation}
If in addition $\tau \neq 0$, there is a unique positive solution $\phi = \phi_0 = \mathrm{const}$, where
\[ \phi_0 = \left| \frac{3b}{\tau} \right|^\frac{1}{6}. \]
If, in turn, $\tau = 0$, we get from Eq.\ (\ref{S1T2constt})
\[ \phi^{\prime \prime} = - \frac{3}{4}b^2 \phi^{-7}. \]
It follows from the maximum principle that the above equation has no positive solution for $b \neq 0$.

Figures \ref{S1T2b1a12} and \ref{S1T2b1a1} show examples of solutions of system (\ref{systemS1T2}) obtained for $\tau = 1 + a \cos(\alpha)$. Figure \ref{S1T2b1a12} corresponds to $a = 1/2$, $b = 1$ (the same set of parameters, as the one chosen for the solution of system (\ref{systemS1S2}) shown in Fig.\ \ref{S1S2b1a12}). Another solution, obtained assuming $a = b = 1$, is shown in Fig.\ \ref{S1T2b1a1}. In both cases we set the series cutoff parameter $N = 25$. The precision with which Eqs.\ (\ref{systemS1T2}) are satisfied is of the order of $10^{-15}$.

Both systems (\ref{systemS1S2}) and (\ref{systemS1T2}) can be easily generalized to include the cosmological constant. System (\ref{systemS1S2}) ($M = \mathbb S^1 \times \mathbb S^2$) with the cosmological constant reads
\begin{subequations}\label{systemS1S2lambda}
\begin{align}
\phi^{\prime \prime} - \frac{1}{4} \phi + \frac{3}{4} b^2 \phi^{-7} - b w \phi^{-1} - \left( \frac{\tau^2}{12} - \frac{1}{3} w^2 - \frac{1}{4} \Lambda \right) \phi^5 &= 0, \label{systemS1S2hamlambda} \\
w^\prime - \frac{1}{2} \tau^\prime + 6 \phi^{-1} w \phi^\prime &= 0. \label{systemS1S2momlambda}
\end{align}
\end{subequations}
System (\ref{systemS1T2}) ($M = \mathbb S^1 \times \mathbb T^2$) can be generalized as
\begin{subequations}\label{systemS1T2lambda}
\begin{align}
\phi^{\prime \prime} + \frac{3}{4} b^2 \phi^{-7} - b w \phi^{-1} - \left( \frac{\tau^2}{12} - \frac{1}{3} w^2 - \frac{1}{4} \Lambda \right) \phi^5 &= 0, \label{systemS1T2hamlambda} \\
w^\prime - \frac{1}{2} \tau^\prime + 6 \phi^{-1} w \phi^\prime &= 0. \label{systemS1T2momlambda}
\end{align}
\end{subequations}

While numerical solutions of the above systems can be easily obtained using our method, the properties of these systems can be remotely different, especially for large $\Lambda > 0$. This can be seen immediately by inspecting the cases with $\tau = \mathrm{const}$ and $w = 0$. In this case system (\ref{systemS1S2lambda}) yields
\[ \phi^{\prime \prime} - \frac{1}{4} \phi + \frac{3}{4} b^2 \phi^{-7} - \left( \frac{\tau^2}{12} - \frac{1}{4} \Lambda \right) \phi^5 = 0. \]
Assume that $b \neq 0$. If $\frac{1}{3}\tau^2 - \Lambda \ge 0$, there exists a unique positive solution. This solution is constant $\phi = \phi_0 = \mathrm{const}$, and it can be obtained as a solution to the cubic equation
\[ - x^2 + 3 b^2 - \left( \frac{\tau^2}{3} - \Lambda \right) x^3 = 0, \]
where $x = \phi_0^4$. If in turn $\frac{1}{3}\tau^2 - \Lambda < 0$, solutions are no longer unique. This case was analyzed in detail in \cite{chrusciel_gicquaud}, where a corresponding bifurcation structure of solutions was also described.

System (\ref{systemS1T2lambda}) with $\tau = \mathrm{const}$ and $w = 0$ yields
\begin{equation}
\label{S1T2lambdaw0}
\phi^{\prime \prime} + \frac{3}{4} b^2 \phi^{-7} - \left( \frac{\tau^2}{12} - \frac{1}{4} \Lambda \right) \phi^5 = 0.
\end{equation}
Assume again that $b \neq 0$. If $\frac{1}{3}\tau^2 - \Lambda > 0$, there exist a unique positive solution $\phi_0$. It is constant, and it is given by
\[ \phi_0 = \left( \frac{3 b^2}{\frac{\tau^2}{3} - \Lambda} \right)^\frac{1}{12}. \]
If $\frac{1}{3}\tau^2 - \Lambda \le 0$, Eq.\ (\ref{S1T2lambdaw0}) has no positive solutions (this can again be easily seen from the maximum principle).

\section*{Acknowledgments}
The authors would like to thank the referees for very useful comments and suggestions. This research was partially supported by the National Natural Science Foundation of China grant No.\ 11671089 and the Polish National Science Centre grant No.\ 2017/26/A/ST2/00530. Part of this work was done while PM was visiting the School of Mathematical Sciences, Fudan University. He would like to thank this institution for the hospitality and financial support.


\end{document}